\documentclass{sig-alternate}

\pdfoutput=1
\usepackage{times}
\usepackage{latexsym}
\usepackage{booktabs}
\usepackage{amsmath,amssymb}
\usepackage[boxed]{algorithm}
\usepackage{algorithmic}
\usepackage{graphicx}
\usepackage{multirow}
\usepackage[tight]{subfigure}

\usepackage[pdftex,naturalnames,pdfpagemode={UseOutlines},letterpaper,bookmarks,bookmarksopen=true,pdfstartview={FitH},bookmarksnumbered=true,pdftitle={-},pdfauthor={-}]{hyperref}

\renewcommand{\paragraph}[1]{\vspace{1mm}\noindent{\bf #1}}
\newcommand{\eat}[1]{}

\newtheorem{definition}{Definition}[section]

\newtheorem{lemma}[definition]{Lemma}

\newtheorem{theorem}[definition]{Theorem}

\newtheorem{example}[definition]{Example}

\begin{document}

\title{Vision Paper: Towards an Understanding of the Limits of Map-Reduce Computation}

\numberofauthors{1}
\author{Foto N. Afrati$^{\dag}$, Anish Das Sarma$^{\sharp}$, Semih Salihoglu$^{\ddag}$, Jeffrey D. Ullman$^{\ddag}$\\
\affaddr{\large $^{\dag}$ National Technical University of Athens,  $^{\sharp}$ Google Research, $^{\ddag}$ Stanford University}\\
\email{afrati@softlab.ece.ntua.gr, anish.dassarma@gmail.com, semih@cs.stanford.edu, ullman@gmail.com}}

\date{}

\maketitle

\section{Introduction}

A significant amount of recent research work has addressed the problem of solving various data management problems in the cloud. The major algorithmic challenges in map-reduce computations involve balancing a multitude of factors such as the number of machines available for mappers/reducers, their memory requirements, and {\em communication cost} (total amount of data sent from mappers to reducers). Most past work provides custom solutions to specific problems, e.g., performing fuzzy joins in map-reduce~\cite{ADMPU12, VCL10}, clustering~\cite{TTLKF11}, graph analyses~\cite{AFU12, SV11, Schank07}, and so on. While some problems are amenable to very efficient map-reduce algorithms, some other problems do not lend themselves to a natural distribution, and have provable lower bounds. Clearly, the ease of ``map-reducability'' is closely related to whether the problem can be partitioned into independent pieces, which are distributed across mappers/reducers. What makes a problem distributable? Can we characterize general properties of problems that determine how easy or hard it is to find efficient map-reduce algorithms?

This is a vision paper that attempts to answer the questions described above. We define and study {\em replication rate}.  Informally, the replication rate of any map-reduce algorithm gives the average number of reducers each input is sent to. There are many ways to implement nontrivial problems in a round of map-reduce; the more parallelism you want, the more overhead you face due to having to replicate inputs to many reducers.  In this paper:

\begin{itemize}

\item We offer a simple model of how inputs and outputs are related, enabling us to study the replication rate of problems.  We show how our model can capture a varied set of problems. (Section~\ref{model-sect})

\item We study two interesting problems---{\em Hamming Distance-1} (Section~\ref{hd1-sect}) and {\em triangle finding} (Section~\ref{other-results})---and show in each case
there is a lower bound on the replication rate that grows as the number of inputs
per reducer shrinks (and therefore as the parallelism grows).
Moreover, we present methods of mapping inputs to reducers that meet
these lower bounds for various values of  inputs/reducer. 

\end{itemize}

\noindent It is our long-term goal to understand how the structure of a problem, as reflected by the input-output relationship in our model, affects the degree of parallelism/replication tradeoff.

\section{The Model}
\label{model-sect}

The model looks simple~-- perhaps too simple.  But with it we can discover some quite interesting and realistic insights into the range of possible map-reduce algorithms for a problem.  For our purposes, a {\em problem} consists of:

\begin{enumerate}

\item
Sets of {\em inputs} and  {\em outputs}.

\item
A {\em mapping} from outputs to sets of inputs.  The intent is that each output depends on only the set of inputs it is mapped to.

\end{enumerate}

\noindent Note that our model essentially captures the notion {\em provenance}~\cite{provenance}. In our context, there are two nonobvious points about this model:

\begin{itemize}

\item
Inputs and outputs are hypothetical, in the sense that they are all the possible inputs or outputs that might be present in an instance of the problem.  Any {\em instance} of the problem will have a subset of the inputs.  We assume that an output is never made unless at least one of its inputs is present, and in many problems, we only want to make the output if {\em all} of its associated inputs are present.

\item
We need to limit ourselves to finite sets of inputs and outputs.  Thus, a finite domain or domains from which inputs and outputs are constructed is often an integral part of the problem statement, and a ``problem'' is really a family of problems, one for each choice of finite domain(s).

\end{itemize}
We hope a few examples will make these ideas clear.

\subsection{Examples of Problems}
\label{prob-ex-subsect}

\begin{example}
\label{join-ex}
Consider the natural join of relations $R(A,B)$ and $S(B,C)$.  The inputs are tuples in $R$ or $S$, and the outputs are tuples with schema $(A,B,C)$.  To make this problem finite, we need to assume finite domains for attributes $A$, $B$, and $C$; say there are $N_A$, $N_B$, and $N_C$ members of these domains, respectively.

Then there are $N_AN_BN_C$ outputs, each corresponding to a triple $(a,b,c)$.  This output is mapped to the set of two inputs.  One is the tuple $R(a,b)$ from relation $R$ and the other is the tuple $S(b,c)$ from relation $S$.  The number of inputs is $N_AN_B + N_BN_C$.

Notice that in an instance of the join problem, not all the inputs will be present.  That is, the relations $R$ and $S$ will be subsets of all the possible tuples, and the output will be those triples $(a,b,c)$ such that both $R(a,b)$ and $S(b,c)$ are actually present in the input instance.
\end{example}

\begin{example}
\label{triangle-ex}
For another example, consider finding tri\-angles.  We are given a graph as input and want to find all  triples of nodes such that in the graph there are edges between each pair of these three nodes.  To model this problem, we need to assume a domain for the nodes of the input graph with $N$ nodes.  An output is thus a set of three nodes, and an input is a set of two nodes.  The output $\{u,v,w\}$ is mapped to the set of three inputs $\{u,v\}$, $\{u,w\}$, and $\{v,w\}$.  Notice that, unlike the previous and next examples, here, an output is a set of more than two inputs.  In an instance of the triangles problem, some of the possible edges will be present, and the outputs produced will be those such that all three edges to which the output is mapped are present.
\end{example}

\begin{example}
\label{hd1-ex}
This example is a very simple case of a similarity join.  The inputs are binary strings, and since we have to make things finite, we shall assume that these strings have a fixed length b.  There are thus $2^b$ inputs.  The outputs are pairs of inputs that are at Hamming distance 1; that is, the inputs differ in exactly one bit.  There are thus $(b/2)2^b$ outputs, since each of the $2^b$ inputs is Hamming distance 1 from exactly $b$ other inputs~-- those that differ in exactly one of the $b$ bits.  However, that observation counts every pair of inputs at distance 1 twice, which is why we must divide by 2.
\end{example}

\begin{example}
\label{sum-ex}
Suppose we have a relation $R(A,B)$ and we want to implement  group-by-and-sum:

\begin{verbatim}
    SELECT A, SUM(B)
    FROM R
    GROUP BY A;
\end{verbatim}
We must assume finite domains for $A$ and $B$. An output is a value of $A$, say $a$, chosen from the finite domain of $A$-values, together with the sum of all the $B$-values.  This output is associated with a large set of inputs: all tuples with $A$-value $a$ and any $B$-value from the finite domain of $B$.  In any instance of this problem, we do not expect that all these tuples will be present, but as long as at least one of them is present, there will be an output for this value $a$.
\end{example}

\subsection{Mapping Schemas and Replication Rate}
\label{rr-subsect}

For many problems, there is a tradeoff between the number of reducers to which a given input must be sent and the number of inputs that can be sent to one reducer.  It can be argued that the existence of such a tradeoff is tantamount to the problem being ``not embarrassingly parallel''; that is, the more parallelism we introduce, the greater will be the total cost of computation.

The more reducers that receive a given input, the greater the communication cost for solving an instance of a problem using map-reduce.  As communication tends to be expensive, and in fact is often the dominant cost, we'd like to keep the number of reducers per input low.  However, there is also a good reason to want to keep the number of inputs per reducer low.  Doing so makes it likely that we can execute the Reduce task in main memory.  Also, the smaller the input to each reducer, the more parallelism there can be and the lower will be the wall-clock time for executing the map-reduce job (assuming there is an adequate number of compute-nodes to execute all the Reduce tasks in parallel).

In our discussion, we shall use the convention that $p$ is the number of reducers used to solve a given problem instance, and $q$ is the maximum number of inputs that can be sent to any one reducer.  We should understand that $q$ counts the number of potential inputs, regardless of which inputs are actually present for an instance of the problem.  However, on the assumption that inputs are chosen independently with fixed probability, we can expect the number of actual inputs at a reducer to be $q$ times that probability, and there is a vanishingly small chance of significant deviation for large $q$.  If we know the probability of an input being present in the data is $x$, and we can tolerate $q_1$ real inputs at a reducer, then we can use $q=q_1/x$ to account for the fact that not all inputs will actually be present.

With this motivation in mind, let us define a {\em mapping schema} for a given problem, with a given value of $q$, to be an assignment of a set of reducers to each input, subject to the constraints that:

\begin{enumerate}

\item
No more than $q$ inputs are assigned to any one reducer.

\item
For every output, its associated inputs are all assigned to one reducer.  We say the reducer {\em covers} the output.  This reducer need not be unique, and it is, of course, permitted that these same inputs are assigned also to other reducers.

\end{enumerate}

The figure of merit for a mapping schema is the {\em replication rate}, which we define to be the average number of reducers to which an input is mapped by that schema.  Suppose that for a certain algorithm, the $i$th reducer is assigned $q_i \le q$ inputs, and let $I$ be the number of different inputs.  Then the replication rate $r$ for this algorithm is
$$r = \sum_{i=1}^p q_i/I$$

We want to derive lower bounds on $r$, as a function of $q$, for various problems, thus demonstrating the tradeoff between high parallelism (many small reducers) and overhead (total communication cost~-- the replication rate).  These lower bounds depend on counting the total number of outputs that a reducer can cover if it is given at most $q$ inputs. We let $g(q)$ denote this number of outputs that a reducer with $q$ inputs can cover.

Observe that, no matter what random set of inputs is present for an instance of the problem, the expected communication is $r$ times the number of inputs actually present, so $r$ is a good measure of the communication cost incurred during an instance of the problem.  Further, the assumption that the mapping schema assigns inputs to processors without reference to what inputs are actually present captures the nature of a map-reduce computation.  Normally, a map function turns input objects into key-value pairs independently, without knowing what else is in the input.

\section{The Hamming-Distance-1 Problem}
\label{hd1-sect}

We are going to begin our development of the model with the tightest result we can offer.  For the problem of finding pairs of bit strings of length $b$ that are at Hamming distance 1, we have a lower bound on the replication rate $r$ as a function of $q$, the maximum number of inputs assigned to a reducer.  This bound is essentially best possible, as we shall point to a number of mapping schemas that solve the problem and have exactly the replication rate stated in the lower bound.

\subsection{Bounding the Number of Outputs}
\label{output-bound-subsect}

The key to the lower bound on replication rate as a function of $q$ is a tight upper bound on the number of outputs that can be covered by a reducer assigned $q$ inputs.
\begin{lemma}
\label{hd-outputs-lemma}
For the Hamming-distance-1 problem, a reducer that is assigned $q$ inputs can cover no more than $(q/2)\log_2q$ outputs.
\end{lemma}

\begin{proof}
The proof is an induction on $b$, the length of the bit strings in the input.  The basis is $b=1$.  Here, there are only two strings, so $q$ is either 1 or 2.  If $q=1$, the reducer can cover no outputs.  But $(q/2)\log_2q$ is 0 when $q=1$, so the lemma holds in this case.  If $q=2$, the reducer can cover at most one output.  But $(q/2)\log_2q$ is 1 when $q=2$, so again the lemma holds.

Now let us assume the bound for $b$ and consider the case where the inputs consist of strings of length $b+1$.  Let $X$ be a set of $q$ bit strings of length $b+1$.  Let $Y$ be the subset of $X$ consisting of those strings that begin with 0, and let $Z$ be the remaining strings of $X$~-- those that begin with 1.  Suppose $Y$ and $Z$ have $y$ and $z$ members, respectively, so $q=y+z$.

An important observation is that for any string in $Y$, there is at most one string in $Z$ at Hamming distance 1.  That is, if $0w$ is in $Y$, it could be Hamming distance 1 from $1w$ in $Z$, if that string is indeed in $Z$, but there is no other string in $Z$ that could be at Hamming distance 1 from $0w$, since all strings in $Z$ start with 1.  Likewise, each string in $Z$ can be distance 1 from at most one string in $Y$.  Thus, the number of outputs with one string in $Y$ and the other in $Z$ is at most $\min(y,z)$.

So let's count the maximum number of outputs that can have their inputs within $X$.  By the inductive hypothesis, there are at most $(y/2)\log_2y$ outputs both of whose inputs are in $Y$, at most $(z/2)\log_2z$ outputs both of whose inputs are in $Z$, and, by the observation in the paragraph above, at most $\min(y,z)$ outputs with one input in each of $Y$ and $Z$.

Assume without loss of generality that $y\le z$.  Then the maximum number of strings of length $b+1$ that can be covered by a reducer with $q$ inputs is
$$\frac{y}{2} \log_2y + \frac{z}{2} \log_2z + y$$
We must show that this function is at most $(q/2)\log_2q$, or, since $q=y+z$, we need to show

\begin{equation}
\label{yz-eq}
\frac{y}{2} \log_2y + \frac{z}{2} \log_2z + y \le \frac{y+z}{2} \log_2(y+z)
\end{equation}
under the condition that $z\ge y$.

First, observe that when $y=z$, Equation~\ref{yz-eq} holds with equality.  That is, both sides become $y(\log_2y + 1)$.  Next, consider the derivatives, with respect to $z$, of the two sides of Equation~\ref{yz-eq}.  $d/dz$ of the left side is
$$\frac{1}{2}\log_2z + \frac{\log_2e}{2}$$
while the derivative of the right side is
$$\frac{1}{2}\log_2(y+z) + \frac{\log_2e}{2}$$
Since $z\ge y\ge0$, the derivative of the left side is always less than or equal to the derivative of the right side.  Thus, as $z$ grows larger than $y$, the left side remains no greater than the right.  That proves the induction step, and we may conclude the lemma.
\end{proof}

\subsection{The Tradeoff for Hamming Distance 1}
\label{hd-trade-subsect}

We can use Lemma~\ref{hd-outputs-lemma} to get a lower bound on the replication rate as a function of $q$, the maximum number of inputs at a reducer.

\begin{theorem}
\label{hd-trade-thm}
For the Hamming distance 1 problem with inputs of length $b$, the replication rate $r$ is at least $b/\log_2q$.
\end{theorem}

\begin{proof}
Suppose there are $p$ reducers, each with $\leq q$ inputs.  By Lemma~\ref{hd-outputs-lemma}, there are at most $(q_i/2)\log_2q_i$ outputs covered by reducer $i$.  

The total number of outputs, given that inputs are of length $b$ is $(b/2)2^b$.  Thus, since every output must be covered, and $log_2q \geq log_2q_i$ for all $i$, we have

\begin{equation}
\label{hd-trade-eq}
\sum_{i=1}^p\frac{q_i}{2}\log_2q_i \ge \frac{b}{2}2^b
\end{equation}
\begin{equation}
\label{hd-trade-eq}
\sum_{i=1}^p\frac{q_i}{2}\log_2q \ge \frac{b}{2}2^b
\end{equation}
The replication rate is $r = \sum_{i=1}^pq_i/2^b$, that is, the sum of the inputs at each reducer divided by the total number of inputs.  We can move factors in Equation~\ref{hd-trade-eq} to get a lower bound on $r = \sum_{i=1}^pq_i/2^b \ge b/\log_2q$, which is exactly the statement of the theorem.
\end{proof}

\subsection{Upper Bound for Hamming Distance 1}
\label{hd-upper-subsect}

There are a number of algorithms for finding pairs at Hamming distance 1 that match the lower bound of Theorem~\ref{hd-trade-thm}.  First, suppose $q=2$; that is, every reducer gets exactly 2 inputs, and is therefore responsible for exactly one output.  Theorem~\ref{hd-trade-thm} says the replication rate $r$ must be at least $b/\log_22 = b$.  But in this case, every input string $w$ of length $b$ must be sent to exactly $b$ reducers~-- the reducers corresponding to that input and the $b$ inputs that are Hamming distance 1 from $w$.

There is another simple case at the other extreme.  If $q = 2^b$, then we need only one reducer, which gets all the inputs.  In that case, $r=1$.  But Theorem~\ref{hd-trade-thm} says that $r$ must be at least $b/\log_2(2^b) = 1$.

In \cite{ADMPU12}, there is an algorithm called Splitting that, for the case of Hamming distance 1 uses $2^{1+b/2}$ reducers, for some even $b$.  Half of these reducers, or $2^{b/2}$ reducers correspond to the $2^{b/2}$ possible bit strings that may be the first half of an input string.  Call these {\em Group I reducers}. The second half of the reducers correspond to the $2^{b/2}$ bit strings that may be the second half of an input.   Call these {\em Group II reducers}.  Thus, each bit string of length $b/2$ corresponds to two different reducers.

An input $w$ of length $b$ is sent to 2 reducers: the Group-I reducer that corresponds to its first $b/2$ bits, and the Group-II reducer that corresponds to its last $b/2$ bits.  Thus, each input is assigned to two reducers, and the replication rate is 2.  That also matches the lower bound of $b/\log_2(2^{b/2}) = b/(b/2) = 2$.  It is easy to observe that every pair of inputs at distance 1 is sent to some reducer in common.  These inputs must either agree in the first half of their bits, in which case they are sent to the same Group-I reducer, or they agree on the last half of their bits, in which case they are sent to the same Group-II reducer.

We can generalize the Splitting Algorithm to give us an algorithm whose replication rate $r$ matches the lower bound, for any integer $r>2$.  We must assume that $r$ divides $b$ evenly.  Thus, strings of length $b$ can be split into $r$ pieces, each of length $b/r$.  We will have $r$ groups of reducers, numbered 1 through $r$.  In each group of reducers there is a reducer corresponding to each of the $2^{b-b/r}$ bit strings of length $b-b/r$.

To see how inputs are assigned to reducers, suppose $w$ is a bit string of length $b$.   Write $w = w_1w_2\cdots w_r$, where each $w_i$ is of length $b/r$.  We send $w$ to the group-$i$ reducer that corresponds to bit string $w_1\cdots w_{i-1}w_{i+1}\cdots w_r$, that is, $w$ with the $i$th substring $w_i$ removed.   Thus, each input is sent to $r$ reducers, one in each of the $r$ groups, and the replication rate is $r$.  The input size for each reducer is $q = 2^{b/r}$, so the lower bound says that the replication rate must be at least $b/\log_2(2^{b/r}) = b/(b/r) = r$.  That is the replication rate of our generalization of the Splitting algorithm is tight.

Finally, we need to argue that the mapping schema solves the problem.  Any two strings at Hamming distance 1 will disagree in only one of the $r$ segments of length $b/r$.  If they disagree in the $i$th segments, then they will be sent to the same Group~$i$ reducer, because reducers in this group ignore the value in the $i$th segment.  Thus, this reducer will cover the output consisting of this pair.

\begin{figure}
\centerline{\includegraphics[width=0.4\textwidth]{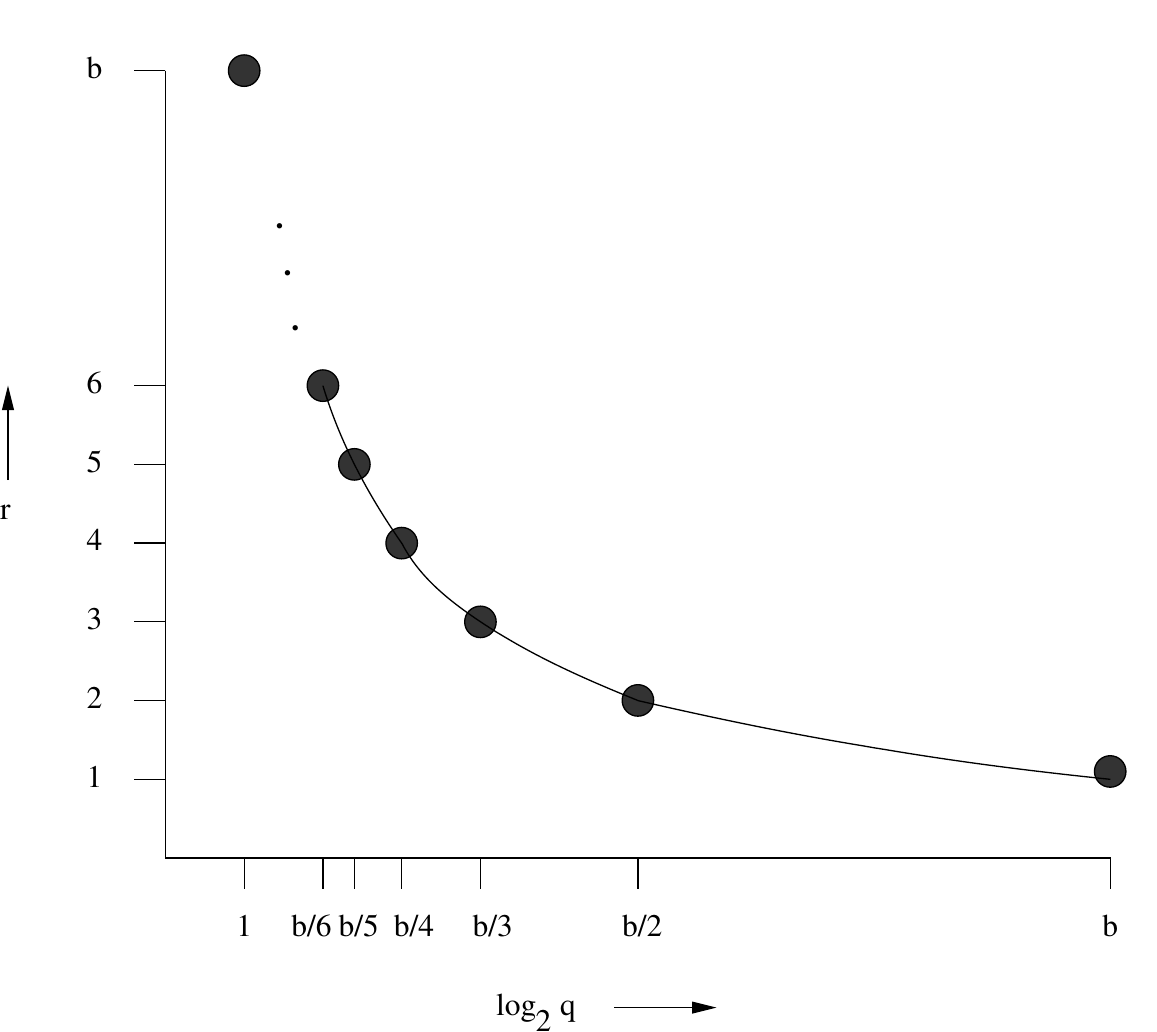}}
\caption{Known algorithms matching the lower bound on replication rate}
\label{tradeoff-fig}
\end{figure}

Figure~\ref{tradeoff-fig} illustrates what we know.  The hyperbola is the lower bound.  Known algorithms that match the lower bound on replication rate are shown with dots.

\pagebreak
\subsection{An Algorithm for Large {\it \large q}}
\label{weights-subsect}

There is a family of algorithms that use reducers with large input~-- $q$ well above $2^{b/2}$, but lower that $2^b$.   The simplest version of these algorithms divides bit strings of length $b$ into left and right halves of length $b/2$ and organizes them by weights, as suggested by Fig.~\ref{weights-fig}.  The {\em weight} of a bit string is the number of 1's in that string.  In detail, for some $k$, which we assume divides $b/2$, we partition the weights into $b/(2k)$ groups, each with $k$ consecutive weights.  Thus, the first group is weights 0 through $k-1$, the second is weights $k$ through $2k-1$, and so on.  The last group has an extra weight, $b/2$, and consists of weights $\frac{b}{2}-k$ through $b/2$.

\begin{figure}
\centerline{\includegraphics[width=0.4\textwidth]{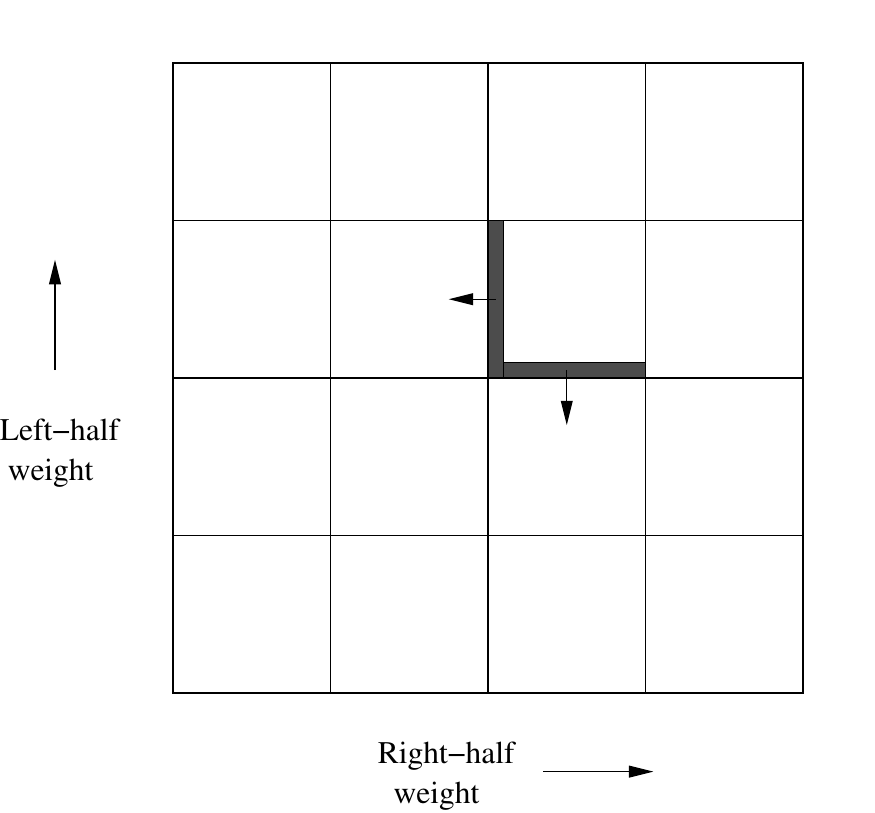}}
\caption{Partitioning by weight.  Only the border weights need to be replicated}
\label{weights-fig}
\end{figure}

There are $(\frac{b}{2k})^2$ reducers; each corresponds to a range of weights for the first half and a range of weights for the second half.  A string is assigned to reducer $(i,j)$, for $i,j=1,2,\ldots, b/2k$  if the left half of the string has weight in the range $(i-1)k$ through $ik-1$ and the  right half of the string has weight in the range $(j-1)k$ through $jk-1$.

Consider two bit strings $w_0$ and $w_1$ of length $b$ that differ in exactly one bit .  Suppose the bit in which they differ is in the left half, and suppose that $w_1$ has a 1 in that bit.  Finally, let $w_1$ be assigned to reducer $R$.  Then unless the weight of the left half of $w_1$ is the lowest weight for the left half that is assigned to reducer $R$, $w_0$ will also be at $R$, and therefore $R$ will cover the pair $\{w_0,w_1\}$.  However, if the weight of $w_1$ in its left half is the lowest possible left-half weight for $R$, then $w_0$ will be assigned to the reducer with the same range for the right half, but the next lower range for the left half.  Therefore, to make sure that $w_0$ and $w_1$ share a reducer, we need to replicate $w_1$ at the neighboring reducer that handles $w_0$.  The same problem occurs if $w_0$ and $w_1$ differ in the right half, so any string whose right half has the lowest possible weight in its range also has to be replicated at a neighboring reducer.  We suggested in Fig.~\ref{weights-fig} how the strings with weights at the two lower borders of the ranges for a reducer need to be replicated at a neighboring reducer.

Now, let us analyze the situation, including the maximum number $q$ of inputs assigned to a reducer, and the replication rate.  For the bound on $q$, note that the vast majority of the bit strings of length $n$ have weight close to $n/2$.  The number of bit strings of weight exactly $n/2$ is $\binom{n}{n/2}$.  Stirling's approximation \cite{feller68} gives us $2^n/\sqrt{2\pi n}$ for this quantity.  That is, one in $O(\sqrt{n})$ of the strings have the average weight.

If we partition strings as suggested by Fig.~\ref{weights-fig}, then the most populous $k\times k$ cell, the one that contains strings with weight $b/4$ in the first half and also weight $b/4$ in the second half, will have no more than
$$k^2\Bigr(\frac{2^{b/2}}{\sqrt{2\pi(b/2)}}\Bigl)^2 = \frac{k^22^b}{\pi b}$$
strings assigned.\footnote{Note that many of the cells have many fewer strings assigned, and in fact a fraction close to 1 of the strings have weights within $\sqrt{b}$ of $b/4$ in both their left halves and right halves.  In a realistic implementation, we would probably want to combine the cells with relatively small population at a single reducer, in order to equalize the work at each reducer.}
If $k$ is a constant, then in terms of the horizontal axis in Fig.~\ref{tradeoff-fig}, this algorithm has $\log_2q$ equal to $b - \log_2b$ plus or minus a constant.  It is thus very close to the right end, but not exactly at the right end.

For the replication rate of the algorithm, if $k$ is a constant, then for any cell there is only a small ratio of variation between the numbers of strings with weights $i$ and $j$ in the left and right halves, for any $i$ and $j$ that are assigned to that cell.  Moreover, when we look at the total number of strings in the borders of all the cells, the differences average out so the total number of replicated strings is very close to $(2k)/k^2 = 2/k$.  That is, a string is replicated if either its left half has a weight divisible by $k$ or its right half does.   Note that strings in the lower-left corner of a cell are replicated twice, strings of the other $2k-2$ points on the border are replicated once, and the majority of strings are not replicated at all.  We conclude that the replication rate is $1+\frac{2}{k}$.

\subsection{Generalization to {\large \tt d} Dimensions}
\label{weights-d-subsect}

The algorithm of Section~\ref{weights-subsect} can be generalized from 2 dimensions to $d$.  Break bit strings of length $b$ into $d$ pieces of length $b/d$, where we assume $d$ divides $b$.  Each string of length $b$ can thus be assigned to a cell in a $d$-dimensional hypercube, based on the weights of each of its $d$ pieces.  Assume that each cell has side $k$ in each dimension, where $k$ is a constant that divides $b/d$.

The most populous cell will be the one that contains strings where each of its $d$ pieces has weight $b/(2d)$.  Again using Stirling's approximation, the number of strings assigned to this cell is
$$k^d\Bigl(\frac{2^{b/d}}{\sqrt{2\pi b/d}}\Bigr)^d = \frac{k^d2^b}{b^{d/2}(2\pi/d)^{d/2}}$$
On the assumption that $k$ is constant, the value of $\log_2q$ is
$$b-(d/2)\log_2b$$
plus or minus a constant.

To compute the replication rate, observe that every point on each of the $d$ faces of the hypercube that are at the low ends of their dimension must be replicated.  The number of points on one face is $k^{d-1}$, so the sum of the volumes of the faces is $dk^{d-1}$.  The entire volume of a cell is $k^d$, so the fraction of points that are replicated is $d/k$, and the replication rate is $1+d/k$.  Technically, we must prove that the points on the border of a cell have, on average, the same number of strings as other points in the cell.  As in Section~\ref{weights-subsect}, the border points in any dimension are those whose corresponding substring has a weight divisible by $k$.  As long as $k$ is much smaller than $b/d$, this number is close to $1/k$th of all the strings of that length.

\section{Triangle Finding}
\label{other-results}

In this section, we present a brief description of other results obtained using our framework, specifically on finding triangles.  The pattern that lets us investigate any problem  is, we hope, clear from the analysis of Section~\ref{hd1-sect}.

\begin{enumerate}

\item
Find an upper bound, $g(q)$, on the number of outputs a reducer can cover if $q$ is the number of inputs it is given.

\item
Count the total numbers of inputs $|I|$ and outputs $|O|$.

\item
Assume there are $p$ reducers, each receiving $q_i \leq q$ inputs and covering $g(q_i)$ outputs.  Together they cover all the outputs. That is $\sum_{i=1}^pg(q_i) \geq |O|$.

\item
Manipulate the inequality from (3) to get a lower bound on
the replication rate, which is $\sum_{i=1}^pq_i / |I|$.

\item
Hopefully, demonstrate that there are algorithms whose replication rate matches the formula from (4).

\end{enumerate}

\subsection{The Tradeoff}
\label{tri-subsect}

We shall briefly show how this method applies to the problem of finding triangles introduced in Example~\ref{triangle-ex}.  Suppose $n$ is the number of nodes of the input graph.  Following the outline just given:

\begin{enumerate}

\item
We claim that the largest number of outputs (triangles) a reducer with at most $q$ inputs occurs when the reducer is assigned all the edges running between some set of $k$ nodes.  This point was proved, to within an order of magnitude in \cite{Schank07}.  Suppose we assign to a reducer all the edges between a set of $k$ nodes.  Then there are $\binom{k}{2}$ edges assigned to this reducer, or approximately $k^2/2$ edges.  Since this quantity is $q$, we have $k = \sqrt{2q}$.
The number of triangles among $k$ nodes is $\binom{k}{3}$, or approximately $k^3/6$ outputs.  In terms of $q$, the upper bound on the number of outputs is $\frac{\sqrt{2}}{3} q^{3/2}$.

\item
The number of inputs is $\binom{n}{2}$ or approximately $n^2/2$.  The number of outputs is $\binom{n}{3}$, or approximately $n^3/6$.

\item
So using the formulas from (1) and (2), if there are p reducers each with $\leq q$ inputs:
$\sum_{i=1}^p\frac{\sqrt{2}}{3} q_i^{3/2} \ge n^3/6$, which implies that $\sum_{i=1}^p\frac{\sqrt{2}}{3} q_iq^{1/2} \ge n^3/6$.

\item
The replication rate is $\sum_{i=1}^pq_i$ divided by the number of inputs, which is $n^2/2$ from (1).  We can manipulate the inequality from (3) to get
$$r = \frac{2\sum_{i=1}^p q_i}{n^2} \ge \frac{n}{\sqrt{2q}}$$

\item
There are known algorithms that, to within a constant factor, match the lower bound on replication rate.  See \cite{SV11} and \cite{AFU12}.\footnote{It is a little tricky to relate these algorithms to the bound, since those algorithms assume the actual data graphs are sparse and calculate replication and input sizes in terms of the number of edges rather than nodes.  However, on randomly chosen subsets of all possible edges, they do get us within a constant factor of the lower bound.}

\end{enumerate}

{\bf Generalizing to multiway joins} Finding triangles is equivalent to computing
the multiway join $E(A,B) \& E(B,C) \& E(C,A)$.
Similar techniques can be used to compute lower and upper bounds for any multiway join.
In particular, in the case where we have one relation of arity $a$ and the multiway join uses $m$ variables then we get lower and upper bounds that are both
$O(q^{1-m/a}n^{m-a})$.


\section{Summary}
\label{summary-sect}

This abstract introduced a simple model for defining map-reduce problems, enabling us to study their ``distributability'' properties. We studied the notion of {\em replication rate}, which is closely related communication cost, and the number of machines available for mappers and reducers. We showed that our model effectively captures a multitude of map-reduce problems, and is a natural formalism for the study of replication rate. We presented a detailed treatment of the hamming-distance-1 problem, providing tight bounds on the replication rate. We also presented a summary of some other results on multiway joins and triangle finding we have obtained.

We believe that our formalism presents a new direction for the study of a large class of map-reduce problems, and allows us to prove results on the limits of map-reducibility for any algorithm for a problem that fits our model.


{\footnotesize
\bibliographystyle{plain}
\bibliography{bib}
}


\clearpage

\end{document}